\documentclass{tPRS2e}
\usepackage{graphicx}
\usepackage{amsmath}
\usepackage{float}
\usepackage{array}
\usepackage{hyperref}
\usepackage{natbib}
\usepackage[utf8]{inputenc}
\usepackage{hhline}
\usepackage{color, colortbl}
\usepackage[center]{caption}

\newtheorem{theorem}{Theorem}
\newtheorem{prop}{Proposition}

\begin{document}

\title{On NP-completeness of the cell formation problem}

\author{
\name{Mikhail V. Batsyn\textsuperscript{a}$^{\ast}$\thanks{$^\ast$Corresponding author. Email: mbatsyn@hse.ru},
Ekaterina K. Batsyna\textsuperscript{b},
Ilya S. Bychkov\textsuperscript{a}
}
\affil{\textsuperscript{a}Laboratory of Algorithms and Technologies for Network Analysis\\
\textsuperscript{b}Department of Applied Mathematics and Informatics\\
National Research University Higher School of Economics, \\25/12 B. Pecherskaya, Nizhny Novgorod, Russian Federation, 603155}
\received{v1.1 released Dec 2018}
}

\maketitle

\begin{abstract}
In the current paper we provide a proof of NP-completeness for the Cell Formation Problem (CFP) with the fractional grouping efficacy objective.
For this purpose we first consider the CFP with the linear objective minimizing the total number of exceptions and voids.
Following the ideas of \citet{Pinheiro} we show that it is equivalent to the Bicluster Graph Editing Problem (BGEP), which is known to be NP-complete \citep{Amit}.
Then we suggest a reduction of the CFP problem with the linear objective function to the CFP with the grouping efficacy objective.
\end{abstract}


\section{Introduction}
The Cell Formation Problem (CFP) consists in optimal grouping of machines together with parts processed on them into manufacturing cells.
The goal of such a bi-clustering (clustering of both machines and parts) is to minimize the inter-cell movement of parts between different cells during the manufacturing process and to maximize the loading of machines with parts processing inside their cells.
The input to this problem is given by a binary machine-part matrix defining for every machine what parts are processed on it.
In terms of input matrix the objective of the CFP is to partition rows (machines) and columns (parts) of the input matrix into rectangular cells minimizing the number of ones outside cells, called exceptions (representing the inter-cell movements of parts), and minimizing the number of zeroes inside cells, called voids (reflecting the underloading of machines).
An example of the input matrix is shown in Table \ref{ex} and a feasible solution for this instance is shown in Table \ref{sol}.

A number of papers on the CFP are devoted to its simplest formulation, called Machine Partitioning Problem (MPP), in which only machines are clustered into cells and the objective is computed as an explicit function from this partition and machine-part matrix \citep{22, 23, 21}.
Though we are not aware of the proof of NP-completeness for the MPP, we believe it exists in literature.
It is probably present in the PhD thesis of \citet{Ballakur}, judging by the references to this work. Unfortunately we have failed to find it in electronic databases.
Besides \citet{Ghosh} states that the NP-hardness of the MPP can be proved "by a straightforward reduction of the clustering problem \citep{Garey}" to the MPP.

The CFP problem becomes much harder when we want to cluster machines and parts together into biclusters.
In spite of the fact that most papers in last decades consider the CFP in its biclustering formulation, there are no papers providing the proof of its NP status to the best of our knowledge.
There is a big number of papers, where authors just write that the problem is NP-hard \citep{Mak, Goncalves, Chan}.
Other authors including \citet{Tunnukij, Elbenani} state that the CFP is NP-hard citing the paper of \citet{Dimopoulos}. But \citet{Dimopoulos} only mention that "the cell-formation problem is a difficult optimization problem".

Many papers including \citet{James, Chung, Paydar, Solimanpur, Utkina} refer to \citet{Ballakur_Steudel} when writing about the NP-hardness of the CFP.
However \citet{Ballakur_Steudel} present a heuristic for the CFP with different objective functions and do not state anything about the NP status of these CFP formulations.
Finally there are some papers citing \citet{Ballakur} PhD thesis where several CFP formulations are considered.
However this paper is not available in any electronic publication databases.
According to the existing references to this thesis and other papers of Ballakur we can only conclude that he considers the machine partitioning and machine-part partitioning problems with some objective functions, but not with the grouping efficacy function introduced later by \citet{Kumar}.
At the same time the grouping efficacy is currently widely accepted and considered as the best function successfully joining the both objectives of inter-cell part movement minimization and inta-cell machine loading maximization.

In the current paper we provide a proof of NP-completeness for the CFP problem with the fractional grouping efficacy objective.
For this purpose we first consider the CFP with the linear objective minimizing the total number of exceptions and voids.
Following the ideas of \citet{Pinheiro} we show that it is equivalent to the Bicluster Graph Editing Problem (BGEP), which is known to be NP-complete \citep{Amit}.
Then we suggest a reduction of the CFP problem with the linear objective function to the CFP with the grouping efficacy objective.

\section{Problem formulation}
In the CFP we are given $m$ machines, $p$ parts processed on these machines, and $m \times p$ Boolean matrix $A$ in which $a_{ij} = 1$, if machine $i$ processes part $j$ during the production process, and $a_{ij} = 0$ otherwise.
We should cluster both machines and parts into biclusters, called cells, so that for every part we minimize simultaneously the number of processing operations of this part on machines from other cells and the number of machines from the same cell which do not process this part.
Thus we minimize the movement of parts to other cells (inter-cell operations) and maximize the loading of machines with processing operations inside cells (intra-cell operations) during the production process.
In other words we need to choose machine-part cells in matrix $A$, such that the number of ones outside these cells (called exceptions) is minimal possible and at the same time the number of zeroes inside these cells (called voids) is also minimal possible.
The objective function which provides a good combination of these two goals and is widely accepted in literature is the grouping efficacy suggested by \citet{Kumar}:
\begin{equation}
f = \frac{n_1 - e}{n_1 + v} \to \max
\end{equation}
Here $n_1$ is the number of ones in the input matrix, $e$ is the number of exceptions (ones outside cells), $v$ is the number of voids (zeroes inside cells).

Below we present a straightforward fractional programming model for the CFP (\citet{Utkina1}, \citet{Bychkov}).
Since the number of cells cannot be greater than the number of machines and the number of parts, then the maximal possible number of cells is equal to $\min(m, p)$. We denote this value as $c = \min(m, p)$.\\
Decision variables:
\begin{equation}
x_{ik} =
    \begin{cases}
    1 & \text{if machine } i \text{ is assigned to cell } k\\
    0 & \text{otherwise}\\
    \end{cases}
\end{equation}
\begin{equation}
y_{jk} =
    \begin{cases}
    1 & \text{if part } j \text{ is assigned to cell } k\\
    0 & \text{otherwise}\\
    \end{cases}
\end{equation}
\begin{equation}
e = n_1 - \sum_{k=1}^{c} \sum_{i=1}^{m} \sum_{j=1}^{p} a_{ij}x_{ik}y_{jk}
\end{equation}
\begin{equation}
v = \sum_{k=1}^{c} \sum_{i=1}^{m} \sum_{j=1}^{p} (1-a_{ij})x_{ik}y_{jk}
\end{equation}
Objective functions:
\begin{subequations}
  \begin{align}
     \label{f1}
     f_1 = e + v \to \min \\
     \label{f2}
     f_2 = \frac{n_1 - e}{n_1 + v} \to \max
  \end{align}
\end{subequations}
Constraints:
\begin{equation}
\label{re_x}
\sum_{k=1}^{c} x_{ik} = 1 \quad \forall i = 1,\ldots,m
\end{equation}
\begin{equation}
\label{re_y}
\sum_{k=1}^{c} y_{jk} = 1 \quad \forall j = 1,\ldots,p
\end{equation}
Objective function \eqref{f1} minimizes the number of exceptions and voids and objective function \eqref{f2} maximizes the grouping efficacy.
Assignment constraints \eqref{re_x} and \eqref{re_y} provide that all machines and parts are partitioned into disjoint cells.

\begin{table}
\centering
        \begin{tabular}{ c|ccccccc|}
                  \multicolumn{1}{c}{~} & \multicolumn{1}{c}{1} & \multicolumn{1}{c}{2} & \multicolumn{1}{c}{3} & \multicolumn{1}{c}{4} & \multicolumn{1}{c}{5} & \multicolumn{1}{c}{6} & \multicolumn{1}{c}{7} \\
                   \hhline{~-------}
                  1 & 1 & 0 & 1 & 0 & 0 & 1 & 1 \\
                  2 & 1 & 1 & 1 & 0 & 0 & 0 & 0 \\
                  3 & 1 & 0 & 1 & 0 & 1 & 0 & 1 \\
                  4 &1 &1 &0 &1 &1 &0 &1 \\
                  5 &1 &1 &1 &1 &1 &0 &0 \\
                  \hhline{~-------}
        \end{tabular}
        \caption{CFP instance}
\label{ex}
\end{table}

\begin{table}
\centering
        \begin{tabular}{ c|ccccccc|}
                  \multicolumn{1}{c}{~} & \multicolumn{1}{c}{1} & \multicolumn{1}{c}{2} & \multicolumn{1}{c}{3} & \multicolumn{1}{c}{4} & \multicolumn{1}{c}{5} & \multicolumn{1}{c}{6} & \multicolumn{1}{c}{7} \\
                   \hhline{~-------}
                  2 &\cellcolor{yellow}1 &\cellcolor{yellow}1 &\multicolumn{1}{c|}{\cellcolor{yellow}1} &0 &0 &0 &0 \\
                  4 & \cellcolor{yellow}1  & \cellcolor{yellow}1 &\multicolumn{1}{c|}{\cellcolor{yellow}0} &1 &1  &0 &1 \\
                  5 &\cellcolor{yellow}1 &\cellcolor{yellow}1 &\multicolumn{1}{c|}{\cellcolor{yellow}1} &1 &1 &0 &0 \\
                   \hhline{~-----}
                  3 &1 &0 &1 &\multicolumn{1}{|c}{\cellcolor{yellow}0} &\multicolumn{1}{c|}{\cellcolor{yellow}1} &0 &1 \\
                   \hhline{~~~~----}
                  1 &1 &0 &1 &0 &0 &\multicolumn{1}{|c}{\cellcolor{yellow}1} &\multicolumn{1}{c|}{\cellcolor{yellow}1} \\
                  \hhline{~-------}
        \end{tabular}
        \caption{CFP solution}
\label{sol}
\end{table}

\section{NP-completeness}
To prove the NP-completeness of the CFP with linear objective \eqref{f1} we use the Bicluster Graph Editing Problem (BGEP).
The first authors who have noticed the closeness of the CFP and BGEP problems are \citet{Pinheiro}.
They applied it in their exact algorithm for the CFP with the grouping efficacy objective.

\begin{figure}[h]
\centering
\includegraphics[scale=1.0]{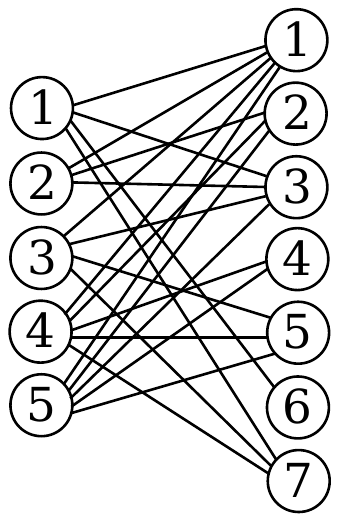}
\caption{BGEP instance}
\label{bgep}
\end{figure}

The BGEP problem consists in determining the minimum number of edges which should be added to/removed from the given bipartite graph so that it transforms to a set of isolated bicliques.
An example of a BGEP instance is presented in Figure \ref{bgep} and its solution -- in Figure \ref{bgep_sol}.
Here dotted thick lines show the added edges and red thin lines -- the removed edges.
The BGEP problem is NP-complete.
To be more exact -- its decision version is NP-complete.
The decision version of an optimization problem with objective function $f \rightarrow \max$ ($f \rightarrow \min$) is a problem with the same constraints, which only answers the question, whether there exists a feasible solution with $f \ge c$ ($f \le c$) for any given constant $c$.
Since the theory of NP-completeness is applicable only for decision problems in all the propositions and theorems below we will talk about the decision versions of the problems.

\begin{figure}[h]
\centering
\includegraphics[scale=1.0]{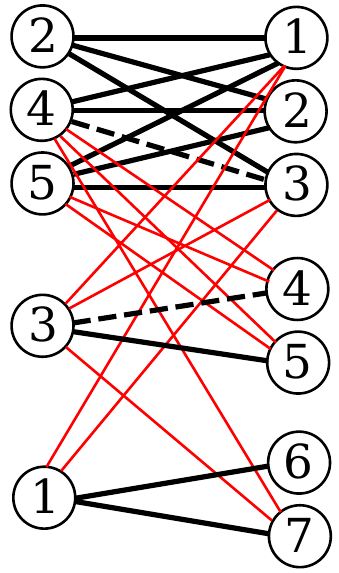}
\caption{BGEP solution}
\label{bgep_sol}
\end{figure}

\begin{theorem}[\citet{Amit}]
The BGEP problem is NP-complete because the NP-complete 3-exact 3-cover problem can be polynomially reduced to BGEP.
\end{theorem}
The 3-exact 3-cover problem is defined as follows. Given a set of elements $U = \{1, 2, ..., 3n\}$ and a collection $C$ of triplets of these elements, such that each element can belong to at most 3 triplets, determine if there exists a subcollection of $C$ with size $n$ which covers $U$.

Hereafter we will call CFP 1 the CFP problem with the linear objective function $f_1 = e + v$ \eqref{f1}, and CFP 2 -- the CFP problem with the grouping efficacy objective $f_2 = (n_1 - e) / (n_1 + v)$ \eqref{f2}.

\begin{theorem}
The CFP with linear objective $f_1 = e + v$ (CFP 1) is NP-complete since it is equivalent to the BGEP problem.
\end{theorem}
\begin{proof}
There is a one-to-one correspondence between these two problems.
Every machine in the CFP corresponds to a vertex in one part of the bipartite graph in the BGEP, and every part in the CFP corresponds to a vertex in another part of this graph.
The machine-part matrix in the CFP coincides with the bipartite graph biadjacency matrix in the BGEP.
Every exception in a solution of the CFP corresponds to an edge which should be removed from the bipartite graph in the BGEP in order to transform it to a set of isolated bicliques.
And every void in a CFP solution corresponds to an edge which should be added to the bipartite graph in the BGEP.

It is clear that the CFP $f_1 = e + v \rightarrow \min$ objective is equivalent to the BGEP objective of minimizing the number of added / removed edges needed to transform the input bipartite graph to a set of isolated bicliques.
Every biclique corresponds to a rectangular cell in the CFP.
If we remove the added edges and return back the removed ones then every isolated clique will become a non-isolated quasi-biclique completely coinciding with a rectangular cell in a CFP solution.
Thus the CFP 1 problem is equivalent to the BGEP problem and it is NP-complete.
\end{proof}

For example, rows (machines) 2, 4, 5, 3, 1 in Table \ref{sol} correspond to vertices 2, 4, 5, 3, 1 in the left part of the bipartite graph in Figure \ref{bgep_sol} and columns (parts) 1, ..., 7 correspond to vertices 1, ..., 7 in the right part of this graph.
The solution of this BGEP instance contains 3 bicliques shown with thick lines in Figure \ref{bgep_sol}.
Here two dashed lines represent two edges which should be added to the graph to form bicliques.
Red thin lines show the edges which should be removed from the graph to isolate the bicliques from each other.

To prove the NP-completeness of the CFP 2 problem we suggest the reduction of CFP 1 problem to it.
The CFP 2 objective can be written in the following way.
\[
f_2 = \frac{n_1 - e}{n_1 + v} = 1 - \frac{e + v}{n_1 + v} \rightarrow \max \quad \Leftrightarrow \quad \frac{e + v}{n_1 + v} \rightarrow \min
\]
This expression is almost equivalent to the linear objective of the CFP 1, except the value of $v$ in the denominator.
Our idea is to nullify the influence of this value by significant increasing of the number of ones $n_1$.
We reduce the CFP 1 problem to CFP 2 by extending the original machine-part matrix $A$ with a big block of ones as it is shown in Table \ref{ext}.
For example, for the CFP 1 instance shown in Table \ref{ex} the extended matrix $\tilde{A}$ will be as shown in Table \ref{ext1}.
Before the main theorem using the suggested reduction and stating the NP-completeness of CFP 2 we will need to prove two propositions first.

\begin{table}
\centering
        \begin{tabular}{|ccc|ccc|}
                  \hhline{------}
                  &    &  & & & \\
                  & A & & & 0 & \\
                  &    & & & & \\
                  \hhline{------}
                  &    & & 1& ... & 1 \\
                  & 0 & & \vdots & & \vdots\\
                  &    & & 1 & ... & 1 \\
                  \hhline{------}
        \end{tabular}
        \caption{Extended matrix $\tilde{A}$}
\label{ext}
\end{table}

\begin{table}
\centering
        \begin{tabular}{ c|cccccccccc|}
                  \multicolumn{1}{c}{~} & \multicolumn{1}{c}{1} & \multicolumn{1}{c}{2} & \multicolumn{1}{c}{3} & \multicolumn{1}{c}{4} & \multicolumn{1}{c}{5} & \multicolumn{1}{c}{6} & \multicolumn{1}{c}{7} & 8 & ... & \multicolumn{1}{c}{42} \\
                   \hhline{~----------}
                  1 & 1 & 0 & 1 & 0 & 0 & 1 & 1 & 0 & ... & 0 \\
                  2 & 1 & 1 & 1 & 0 & 0 & 0 & 0 & 0 & ... & 0 \\
                  3 & 1 & 0 & 1 & 0 & 1 & 0 & 1 & 0 & ... & 0 \\
                  4 & 1 & 1 & 0 & 1 & 1 & 0 & 1 & 0 & ... & 0 \\
                  5 & 1 & 1 & 1 & 1 & 1 & 0 & 0 & 0 & ... & 0 \\
                  6 & 0 & 0 & 0 & 0 & 0 & 0 & 0 & 1 & ... & 1 \\
          \vdots &    &   &   & \vdots  &   &   &  & \vdots & & \vdots\\
                40 & 0 & 0 & 0 & 0 & 0 & 0 & 0 & 1 & ... & 1 \\
                   \hhline{~----------}
        \end{tabular}
        \caption{Extended matrix example}
\label{ext1}
\end{table}

\begin{prop}
If the machine-part matrix for the CFP 2 problem has identical rows then there will be optimal solutions in which these rows belong to the same cell.
\end{prop}
\begin{proof}
Let us assume that there are two identical rows which belong to different cells in an optimal solution, the first of these rows has $e_1$ exceptions (ones outside its cell) and $v_1$ voids (zeroes inside its cell), the second row has $e_2$ exceptions and $v_2$ voids, and all other rows in this solution have in total $e$ exceptions and $v$ voids.
Then the objective function value for this solution is the following.
\[
f = \frac{n_1 - e - e_1 - e_2}{n_1 + v + v_1 + v_2}
\]

If we move the second of the identical rows to the cell of the first one then these two rows will have $2e_1$ exceptions and $2v_1$ voids.
Otherwise, if we move the first row to the cell of the second one, we will get $2e_2$ exceptions and $2v_2$ voids.
Without loss of generality we can assume that joining the identical rows in the cell of the first of them gives the value of the grouping efficacy not smaller than we get in the opposite variant:
\begin{eqnarray*}
\frac{n_1 - e - 2e_1}{n_1 + v + 2v_1} \ge \frac{n_1 - e - 2e_2}{n_1 + v + 2v_2} \quad \Leftrightarrow \\
v_2(n_1 - e) - e_1(n_1 + v) - 2 e_1 v_2  \ge v_1(n_1 - e) - e_2(n_1 + v) - 2 e_2 v_1
\end{eqnarray*}

Now we will prove that the first variant of joining the identical rows gives the objective function value not worse than the original optimal solution has.
We need to prove the following.
\begin{eqnarray*}
\frac{n_1 - e - 2e_1}{n_1 + v + 2v_1} \ge \frac{n_1 - e - e_1 - e_2}{n_1 + v + v_1 + v_2} \quad \Leftrightarrow \\
(v_1 + v_2)(n_1 - e) - 2e_1(n_1 + v) - 2 e_1 (v_1 + v_2)  \ge 2v_1(n_1 - e) - (e_1 + e_2)(n_1 + v) - 2v_1(e_1 + e_2) \quad \Leftrightarrow \\
v_2(n_1 - e) - e_1(n_1 + v) - 2 e_1 v_2  \ge v_1(n_1 - e) - e_2(n_1 + v) - 2 e_2 v_1
\end{eqnarray*}
The last line exactly coincides with the expression we have obtained above from our assumption that the first variant of joining the identical rows is not worse than the second one.
Thus the solution with the joined rows is also optimal.
\end{proof}

The next proposition determines how much ones it is enough to add in the extended matrix in order to nullify the influence of $f_2$ denominator.

\begin{prop}
If the number of added ones $\Delta n_1$ in the extended matrix $\tilde{A}$ is equal to $(mp)^2$ then the maximum of $f_2$ on $\tilde{A}$ is obtained at the same solution (extended with the cell of added ones) at which $f_1$ has its minimum on matrix $A$.
\end{prop}
\begin{proof}
According to Proposition 2 the optimal solution for CFP 2 on the extended matrix has the added block of ones as a separate cell.
This means that this block adds no voids or exceptions to the solution and thus a CFP 1 solution and the corresponding CFP 2 solution (obtained by adding the block of ones as an additional cell) have the same number of voids $v$ and exceptions $e$.
We will now prove that if $\Delta n_1 = (mp)^2$ then for any two CFP 1 solutions with objective function values $f_1$ and $f'_1$ and the correspoding CFP 2 solutions with objective function values $f_2$ and $f'_2$ from $f'_1 < f_1$ it follows that $f'_2 > f_2$.
\begin{eqnarray*}
f'_2 = \frac{\tilde{n}_1 - e'}{\tilde{n}_1 + v'}, \quad f_2 = \frac{\tilde{n}_1 - e}{\tilde{n}_1 + v}, \quad f'_1 = e' + v', \quad f_1 = e + v \\
f'_2 > f_2 \quad \Leftrightarrow \quad \frac{\tilde{n}_1 - e'}{\tilde{n}_1 + v'} > \frac{\tilde{n}_1 - e}{\tilde{n}_1 + v} \quad \Leftrightarrow \quad \frac{e' + v'}{\tilde{n}_1 + v'} < \frac{e + v}{\tilde{n}_1 + v} \quad \Leftrightarrow \quad \\
\frac{f'_1}{\tilde{n}_1 + v'} < \frac{f_1}{\tilde{n}_1 + v} \quad \Leftrightarrow \quad f'_1 < f_1 \frac{\tilde{n}_1 + v'}{\tilde{n}_1 + v} \quad \Leftrightarrow \quad f'_1 < f_1 + f_1 \frac{v' - v}{\tilde{n}_1 + v}
\end{eqnarray*}
Note that in case $n_1 = 0$ the CFP 1 problem becomes trivial and so we consider only the case $n_1 \ge 1$.
Since $f_1 \le mp, v - v' \le mp$ for $\Delta n_1 = (mp)^2$ we have:
\begin{equation*}
f_1 \frac{v - v'}{\tilde{n}_1 + v} = f_1 \frac{v - v'}{n_1 + \Delta n_1 + v} \le \frac{(mp)^2}{(mp)^2 + 1} 
\quad \Leftrightarrow \quad
f_1 \frac{v' - v}{\tilde{n}_1 + v} \ge \frac{-(mp)^2}{(mp)^2 + 1}
\end{equation*}
From this it follows that:
\begin{equation*}
f_1 + f_1 \frac{v' - v}{\tilde{n}_1 + v} \ge f_1 - \frac{(mp)^2}{(mp)^2 + 1} > f_1 - 1
\end{equation*}
Since $f'_1$ and $f_1$ are integer, then $f'_1 < f_1$ is equivalent to $f'_1 \le f_1 - 1$.
Thus we have:
\begin{equation*}
f_1 + f_1 \frac{v' - v}{\tilde{n}_1 + v} > f_1 - 1 \ge f'_1 \quad \Rightarrow \quad f'_1 < f_1 + f_1 \frac{v' - v}{\tilde{n}_1 + v}  \quad \Leftrightarrow \quad  f'_2 > f_2
\end{equation*}
So we get that from $f'_1 < f_1$ it follows that $f'_2 > f_2$.
This means that the minimum value of $f_1$ gives the maximum of $f_2$ on the "extended" solution.
\end{proof}

\begin{theorem}
The CFP with grouping efficacy objective $f_2 = (n_1 - e) / (n_1 + v)$ (CFP 2) is NP-complete because CFP 1 can be polynomially reduced to it.
\end{theorem}
\begin{proof}
We will prove that CFP 1, which answers the question, whether there exists a solution with $f_1 = e + v \le c$ with the input matrix $A$ can be polynomially reduced to problem CFP 2 on the extended matrix $\tilde{A}$ (see Table \ref{ext}), which answers the question, whether there exists a solution with $f_2 = (n_1 - e) / (n_1 + v) \ge \tilde{c}$.
Here constant $\tilde{c}$ can depend on constant $c$ and other input parameters.

According to Proposition 3 to get a better solution for CFP 2 we should simply extend the solution for CFP 1 with an additional cell represented by the added block of ones in the extended matrix $\tilde{A}$ (see Table \ref{ext}).
It is clear that for the considered decision version of CFP 2 we should also take the best possible solution with maximal value of $f_2$ to guarantee the satisfaction of inequality $f_2 \ge \tilde{c}$.
Thus the solution for CFP 1 and the corresponding suggested solution for CFP 2 are connected in the following way.
\begin{equation*}
\tilde{e} = e, \quad \tilde{v} = v, \quad \tilde{n}_1 = n_1 + \Delta n_1
\end{equation*}
\begin{equation*}
f_2 = \frac{\tilde{n}_1 - \tilde{e}}{\tilde{n}_1 + \tilde{v}} = \frac{\tilde{n}_1 - e}{\tilde{n}_1 + v} = 1 - \frac{e + v}{\tilde{n}_1 + v} = 1 - \frac{f_1}{\tilde{n}_1 + v}
\end{equation*}
Let us find a value of $\tilde{c}$ such, that for a CFP 1 solution with $f_1 \le c$ the corresponding solution for CFP 2 will have $f_2 \ge \tilde{c}$.
\begin{equation*}
f_1 \le c \quad \Rightarrow \quad f_2 = 1 - \frac{f_1}{\tilde{n}_1 + v} \ge 1 - \frac{c}{\tilde{n}_1}
\end{equation*}
So for $\tilde{c} = 1 - c / \tilde{n}_1$ we have $f_1 \le c \; \Rightarrow \; f_2 \ge \tilde{c}$.
This guarantees that if there are no solutions for CFP 2 with $f_2 \ge \tilde{c}$ then there exist no solutions for CFP 1 with $f_1 \le c$.

Now let us prove that for this $\tilde{c}$ from $f_2 \ge \tilde{c}$ for CFP 2 solution it follows that the original CFP 1 solution has $f_1 \le c$.
We have:
\begin{eqnarray*}
f_2 = \frac{\tilde{n}_1 - e}{\tilde{n}_1 + v} = 1 - \frac{e + v}{\tilde{n}_1 + v} \ge 1 - c / \tilde{n}_1 \quad \Leftrightarrow \\
\frac{e + v}{\tilde{n}_1 + v} \le c / \tilde{n}_1 \quad \Leftrightarrow \quad e + v \le c + \frac{cv}{\tilde{n}_1}
\end{eqnarray*}
Now we can use the fact that the number of added ones is $\Delta n_1 = (mp)^2$, and so $\tilde{n}_1 = n_1 + (mp)^2$.
We also note that the cases when $n_1 = 0$ or $c \ge mp$ are trivial, because in such cases we do not need to construct any CFP 2 instance and can immediately answer the CFP 1 question.
Since $c < mp$ and $v \le mp$ we have.
\begin{eqnarray*}
e + v \le c + \frac{cv}{\tilde{n}_1} < c + \frac{(mp)^2}{n_1 + (mp)^2} < c + 1
\end{eqnarray*}
Since $e + v$ is integer we can conclude that $f_1 = e + v \le c$.

Thus we have found the value of $\tilde{c} = 1 - c / \tilde{n}_1$ such that the answer for any CFP 1 instance on the question, whether there exists a solution with $f_1 \le c$, is "yes", if and only if the answer for the corresponding CFP 2 instance on the question, whether there exists a solution with $f_2 \ge \tilde{c}$, is also "yes".
Consequently, the answer to the CFP 1 question is "no", if and only if, the answer to the CFP 2 question is also "no".
This proves that CFP 2 is at least as hard as the NP-complete problem CFP 1.
It is also clear that CFP 2 belongs to class NP, because any "yes"-solution can be verified in polynomial time.
This proves that CFP 2 is an NP-complete problem.
\end{proof}

\section*{Funding}
Sections 1, 2 and Theorems 1, 2 in Section 3 were prepared within the framework of the Basic Research Program at the National Research University Higher School of Economics (NRU HSE).
Propositions 1, 2 and Theorem 3 in Section 3 were formulated and proved with the support of RSF grant 14-41-00039.

\end{document}